\documentclass[11pt,leqno]{amsart}

\usepackage{amsmath}
\usepackage{amssymb}
\usepackage{a4wide}
\usepackage{bbm}
\usepackage{graphics}
\usepackage{epsfig}
\usepackage{hyperref}

\hypersetup{%
  colorlinks = true,
  linkcolor  = blue,
  citecolor    = red
}
  \usepackage[applemac]{inputenc}
\usepackage{amsmath, amssymb, latexsym, amscd, amsthm,amsfonts,amstext}
\usepackage[mathscr]{eucal}
\usepackage{appendix}
\usepackage[active]{srcltx}
\include{srctex}

\usepackage{color}

 \usepackage{mleftright}
\usepackage{pgf,tikz}
\usepackage{caption}  
\usepackage{mathrsfs}
\usetikzlibrary{quotes, angles,decorations.pathreplacing}
\usepackage{comment}

\newcounter{hypo}

\definecolor{brightmaroon}{rgb}{0.76, 0.13, 0.28}
\definecolor{britishracinggreen}{rgb}{0.0, 0.26, 0.15}
\definecolor{cadmiumgreen}{rgb}{0.0, 0.42, 0.24}
\definecolor{darkmidnightblue}{rgb}{0.0, 0.2, 0.4}
\definecolor{darkpink}{rgb}{0.91, 0.33, 0.5}
\definecolor{teal}{rgb}{0.0, 0.5, 0.5}
\definecolor{burgundy}{rgb}{0.5, 0.0, 0.13}
\definecolor{azure}{rgb}{0.0, 0.5, 1.0}
\definecolor{ao}{rgb}{0.0, 0.5, 0.0}

\makeatletter
 \@addtoreset{equation}{section}
 \makeatother
 
 \usepackage{hyperref} 

\newtheorem{theorem}{Theorem}[section]
\newtheorem{lemma}[theorem]{Lemma}
\newtheorem{proposition}[theorem]{Proposition}

\newtheorem{remark}[theorem]{Remark}

\theoremstyle{definition} 

\newcommand{\fact}[1]{#1\mathpunct{}!}
\numberwithin{equation}{section}
 
  {\setlength{\baselineskip}{1.5\baselineskip}

\title[Resolvent estimates and resonance free domains for systems]
{Resolvent estimates and resonance free domains for Schr\"odinger operators with matrix-valued potentials}
 \author[M. Assal]{Marouane Assal}
\address{Marouane Assal, Facultad de Matem\'aticas, Pontificia Universidad Cat\'olica de Chile, Vicuna Mackenna 4860, Santiago de Chile}
\email{marouane.assal@mat.uc.cl}

\subjclass[2010]{}
 
\begin{document}
\definecolor{qqqqff}{rgb}{0,0,1}

\maketitle

\begin{abstract}
We establish semiclassical resolvent estimates for Schr\"odinger operators with long-range matrix-valued potentials. As an application we prove resonance free domains both in trapping and non-trapping situations. Our results generalize the well-known results of \cite{Burq1, Martinez1} in the case of scalar Schr\"odinger operators.
\end{abstract}


\section{Introduction and Background}


We are interested in resolvent estimates and quantum resonances for Schr\"odinger operators with long-range matrix-valued potentials in the semiclassical regime. The results established here are a part of a work in progress \cite{Assal1} where we will give some applications to scattering theory for matrix Schr\"odinger operators.

\subsection{Preliminaries} Consider the semiclassical Schr\"odinger operator on the Hilbert space $L^2(\mathbb R^d;\mathbb C^N)$, $d\geq 1$,
\begin{equation}\label{Opprinc}
P(h) := -h^2 \Delta \cdot I_N + V(x),
\end{equation}
where $I_N$ is the identity $N\times N$ matrix and $V: \mathbb R^d\rightarrow \mathcal{H}_N$ is a smooth $N\times N$ hermitian matrix-valued potential, i.e.,
$$
V(x) = (V_{ij}(x))_{1\leq i,j\leq N}, \quad V_{ij}(x) = \overline{V_{ji}(x)},
$$ with long-range behavior at infinity, i.e., there exists a constant matrix $V_{\infty}\in \mathcal{H}_N$ and $\rho_0>0$ such that
\begin{equation}\label{hyp2}
\Vert \partial_x^{\alpha} \big(V(x) - V_{\infty} \big) \Vert_{N\times N} = \mathcal{O}_{\alpha}(\langle x\rangle^{-\rho_0 - \vert\alpha\vert}) , \quad \forall x\in \mathbb R^d, \alpha\in \mathbb N^d,
\end{equation}
with $\langle x\rangle:= (1+ \vert x\vert^2)^{1/2}$. Here $\mathcal{H}_N$ denotes the space of $N\times N$ hermitian matrices endowed with the norm $\Vert \cdot \Vert_{N\times N}$ defined by \eqref{matrixnorm} and $h>0$ is the semiclassical parameter. Such operators arise as important models in molecular physics and quantum chemistry, for instance in the Born-Oppenheimer approximation which allows for a drastic reduction of problem size when dealing with molecular systems. In this case, the semiclassical parameter $h$ represents the square root of the quotient between the electronic and nuclear masses (see e.g. \cite{Klein, Jecko}).

Set $P_{\infty}(h):= -h^2\Delta \cdot I_N + V_{\infty}$ and we denote $\lambda_{1}\leq \cdots \leq \lambda_{N}$ the eigenvalues of the constant matrix $V_{\infty}$. Without loss of generality, we can assume that $V_{\infty}$ is diagonal, i.e., 
$$
V_{\infty} = {\rm diag}\,(\lambda_{1},...,\lambda_{N}).
$$

The operator $P_{\infty}(h)$ is self-adjoint in $L^2(\mathbb R^d;\mathbb C^N)$ with domain the Sobolev space $H^2(\mathbb R^d;\mathbb C^N)$, and its spectrum coincides with $[\lambda_{1},+\infty)$. Since $V-V_{\infty}$ is $\Delta$-compact, it follows by the Weyl perturbation Theorem that the operator $P(h)$ admits a unique self-adjoint realization in $L^2(\mathbb R^d;\mathbb C^N)$ with domain $H^2(\mathbb R^d;\mathbb C^N)$ and for any fixed $h>0$, $\sigma_{{\rm ess}}(P(h)) = \sigma_{{\rm ess}}(P_{\infty}(h)) =  [\lambda_{1},+\infty)$. Thus, the operator $P(h)$ may have discrete eigenvalues in $(-\infty,\lambda_{1})$ and embedded ones in the interval $[\lambda_{1},\lambda_{N}]$.

Let 
$$
R_h(z):= (P(h)-z)^{-1}\in \mathcal{L}(L^{2}(\mathbb R^d;\mathbb C^N)), \;\; \; \Im z\neq 0,
$$
denotes the resolvent of $P(h)$. It follows by the Limiting Absorption Principle (see e.g. \cite{Agmon}), using the dilation generator as a scalar conjugate operator, that for any $E>\Vert V_{\infty}\Vert_{N\times N}$ and any $s>\frac{1}{2}$, the boundary value of the resolvent 
$$
R_h(E\pm i0) := \lim_{\varepsilon \to 0^+} R_h(E\pm i\varepsilon)
$$ 
exists as a bounded operator in $\mathcal{L}(L^{2,s}(\mathbb R^d;\mathbb C^N), L^{2,-s}(\mathbb R^d;\mathbb C^N))$. Here $L^{2,s}(\mathbb R^d;\mathbb C^N)$ denotes the space of $\mathbb C^N$-valued functions that are square integrable on $\mathbb R^d$ with respect to the measure $\langle x \rangle^s dx$, equipped with its natural norm 
$$
\Vert u \Vert_{L^{2,s}(\mathbb R^d;\mathbb C^N)} := \Vert \langle x \rangle^s u \Vert_{L^{2}(\mathbb R^d;\mathbb C^N)} = \left(\int_{\mathbb R^d} \vert u(x)\vert_{\mathbb C^N}^2 \langle x \rangle^{2s} dx\right)^{\frac12}.
$$

A natural question is to study the behavior of $R_h(E\pm i 0)$ as $h\rightarrow 0^+$, more precisely, to estimate the size of the norm 
\begin{equation}\label{main norm}
\Vert R_h(E\pm i 0)  \Vert_{ L^{2,s}(\mathbb R^d;\mathbb C^N)\to L^{2,-s}(\mathbb R^d;\mathbb C^N) }, \;\;\;\; s>\frac12, 
\end{equation}
with respect to $h$. This problem is important in scattering theory, for instance, for studying the behavior of observables like the scattering matrix and the total cross section (see e.g., \cite{Robert-Tamura, Nakamura, Michel}). Moreover, it is well known that the semiclassical behavior of the resolvent near a given energy-level have a deep relationship with the existence or the absence of resonance near this level. As we shall see later, getting estimates on the size of the above norm entails important results on the location of resonances of the operator $P(h)$ near the energy-level $E$. Let ${\rm Res}\,(P(h))$ denote the set of resonances of $P(h)$ which we will define rigourously in Section \ref{Section2}.

\subsection{Background on the scalar case} In the scalar case $N=1$ (and $V_{\infty}=0$), it is well known that the size of \eqref{main norm} is $\mathcal{O}(h^{-1})$ in non-trapping situations, that is  
\begin{equation}\label{entr}
\Vert R_h(E\pm i 0)  \Vert_{ L^{2,s}(\mathbb R^d)\to L^{2,-s}(\mathbb R^d) } \lesssim  h^{-1} , \;\; \forall s>\frac12,
\end{equation}
provided that $E>0$ is non-trapping for the classical Hamiltonian $p(x,\xi):=\xi^2+ V(x)$, $(x,\xi)\in T^*\mathbb R^d$, associated with $P(h)$. We recall that an energy $E>0$ is said to be non-trapping for $p$ if the set of trapped trajectories at $E$ defined by
$$
\mathcal{T}(E):= \left\{(x,\xi)\in p^{-1}(E); {\rm exp}\,(tH_p)(x,\xi) \nrightarrow \infty \; \text{as}\;\; t\rightarrow \pm \infty \right\}
$$
is empty. Here ${\rm exp}\,(tH_p): T^*\mathbb R^d \rightarrow T^*\mathbb R^d$ is the flow generated by the Hamiltonian vector field $H_p= 2 \xi \cdot \partial_x - \nabla_x V \cdot \partial_{\xi}
$. This result was proved by Robert and Tamura \cite{Robert-Tamura} using Mourre theory and Fourier integral methods. In particular, estimate \eqref{entr} is the key ingredient in the estimation of the behavior of the scattering cross-section and the complete asymptotic expansion in powers of $h$ for the spectral shift function associated to the operator pair $(-h^2\Delta + V(x), -h^2 \Delta)$. A shorter proof using the construction of a global escape function and Mourre theory was given later by G\'erard and Martinez \cite{Gerard}. 

Another consequence of the non-trapping hypothesis on the energy $E>0$ is the following absence of resonances result due to Martinez \cite{Martinez1}
$$
{\rm Res} (P(h)) \cap  \bigg\{ z\in \mathbb C;\, {\rm Re}\, z \in [E-\varepsilon_0,E+\varepsilon_0] \,\, {\rm and} \,\, {\rm Im}\, z\geq -C h \vert \ln h \vert  \bigg\} = \emptyset,
$$
for some $\varepsilon_0>0$ and all $C>0$, $h\in (0,h_C]$. Martinez's approach is based on some microlocal weighted estimates combined with the construction of a global escape function associated to the classical Hamiltonian $p$. An alternative approach to conjugated operators was introduced later by Sj\"ostrand and Zworski \cite{Sj-Zw1}.

Now, without any assumption on the set of trapped trajectories (trapping situations), we have the following estimates for any $s>\frac12$ and $E>0$,
\begin{equation}\label{estimation exponentielle du th scalaire} 
 \Vert  R_h(E\pm i0) \Vert_{L^{2,s}(\mathbb R^d)\rightarrow L^{2,-s}(\mathbb R^d)}   \lesssim e^{\frac{C}{h}},
\end{equation}
\begin{equation}\label{estimation du th scalaire intro}
\big\Vert  \mathbbm 1_{\{x\in \mathbb R^d;\vert x\vert\geq R_0\}}R_h(E\pm i0) \mathbbm 1_{\{x\in \mathbb R^d;\vert x\vert\geq R_0\}}   \big\Vert_{L^{2,s}(\mathbb R^d)\rightarrow L^{2,-s}(\mathbb R^d)}  \lesssim  h^{-1},
\end{equation}
for some constants $C,R_0>0$ and $h>0$ small enough. Here $\mathbbm 1_{\{x\in \mathbb R^d;\vert x\vert\geq R_0\}}$ is the characteristic function of the set $\{x\in \mathbb R^d; \vert x\vert\geq R_0\}$. These estimates are originally due to Burq \cite{Burq1} in a more general framework including perturbations of second-order elliptic operators with smooth coefficients, obstacle scattering and metric scattering. Others proofs was given later by Vodev \cite{Vodev} and Sj\"ostrand \cite{Sjostrand}. Cardoso and Vodev \cite{Cardoso} and recently Rodianski and Tao \cite{Tao} established generalizations of these estimates to Schr\"odinger operators on manifolds. Carleman estimates are the main tool in all these works. Recently, Datchev \cite{Datchev} provided an elementary proof of these estimates for scalar Schr\"odinger operators on $\mathbb R^d$, $d\neq 2$ (see \cite{Shapiro} for the case $d=2$). The proof in \cite{Datchev} is based on a global weighted Carleman estimate which has the advantage to be valid under low regularity assumption on the potential $V$ and where the construction of the weight function is simple and explicit. More precisely, Datchev's method only requires that the potential $V$ is bounded on $\mathbb R^d$ together with its radial derivative, and they decay like $(1+\vert x \vert)^{-\rho_0}$ and $(1+\vert x \vert)^{-\rho_0-1}$ at infinity, respectively. 

The $h$-dependence in \eqref{estimation exponentielle du th scalaire} and \eqref{estimation du th scalaire intro}  is optimal in general, that is without any assumption on the underlying classical dynamics. Moreover, estimate \eqref{estimation du th scalaire intro} is not true in general when removing one of the characteristic function $\mathbbm 1_{\{x\in \mathbb R^d;\vert x\vert\geq R_0\}}$ (see \cite{Da-Dy-Zw}).

As a consequence of estimate \eqref{estimation exponentielle du th scalaire}, Burq \cite{Burq1} proved the following resonance free region
$$
{\rm Res} (P(h)) \cap \bigg\{ z\in \mathbb C;\, {\rm Re}\, z \in J \,\, {\rm and} \,\, {\rm Im}\, z\geq -Ce^{-C/h} \bigg\} = \emptyset,
$$
for any compact interval $J\subset (0,+\infty)$ and some constant $C>0$, and $h>0$ small enough.

\subsection{Comments on the matrix-valued case} 

In the matrix-valued case, that is, when the considered operator is of the form \eqref{Opprinc}, the situation is more complicated. Notice that since the eigenvalues are not enough regular in general, the usual definition of the Hamiltonian flow for a matrix-valued Hamiltonian function does not make sense (see \cite{Kato}). The non-trapping condition in this case is usually characterized by the existence of a global escape function associated with the matrix-valued symbol of the operator $P(h)$ at the considered energy-level.

Let $p(x,\xi):= \vert \xi\vert^2 I_N+ V(x)$, $(x,\xi)\in T^*\mathbb R^d$, be the matrix-valued semiclassical symbol of $P(h)$. We denote $\lambda_1(x)\leq \lambda_2(x)\leq \cdots \leq \lambda_N(x)$ the (real) eigenvalues of $V(x)$, $x\in \mathbb R^d$. In general, the functions $x\mapsto\lambda_j(x)$ are continuous on $\mathbb R^d$, $j=1,...,N$. For an energy $E\in \mathbb R$, we denote $\Sigma_E$ the corresponding energy surface defined by
$$
\Sigma_E:=  \bigcup_{j=1}^N\left\{ (x,\xi)\in T^*\mathbb R^d; \, \vert \xi\vert^2+ \lambda_j(x)=E\right\}.
$$
A  smooth real-valued function $G\in C^{\infty}(T^*\mathbb R^{d};\mathbb R)$ is a global escape function associated with the classical Hamiltonian $p$ at $E$ if there exists a constant $c>0$ such that $\{p,G\}_{|\Sigma_{E}}\geq c I_N$ in the sense of hermitian matrices, i.e.,
\begin{equation}\label{la fonction fuite}
\big(\{p,G\}(x,\xi)w,w \big)_{\mathbb C^N}\geq c |w|^2 , \quad \forall (x,\xi)\in \Sigma_{E}, \forall w \in \mathbb C^N.
\end{equation}
Here $\{p,G\}:= \partial_{\xi}p  \cdot \partial_{x}G- \partial_{x}p \cdot \partial_{\xi}G$ denotes the Poisson bracket of $p$ and $G$, and $(\cdot,\cdot)_{\mathbb C^N}$ denotes the inner product in $\mathbb C^N$.

In the scalar case $N=1$, the existence of a global escape function associated with the symbol $p$ at an energy $E>0$ is equivalent to the non-trapping condition (see for instance \cite{Gerard}). In \cite{Jecko1}, under the existence condition of a global escape function, Jecko proved that estimate \eqref{entr} still holds in the case of matrix-valued potentials. With this result at hand, the main challenge consists in the construction of a global escape function which  may be quite complicated. This question has been the subject of many works (see \cite{Jecko2, Jecko1, Jecko3, Clotide} and the references therein) for different type of eigenvalue crossings. 

In this note, we provide generalizations of Burq's estimates \eqref{estimation exponentielle du th scalaire} and \eqref{estimation du th scalaire intro} to the case of Schr\"odinger operators with matrix-valued potentials using the approach developed in \cite{Datchev} and we prove related results on the absence of resonances near the real axis.  We refer to \cite{Assal1} for applications of these results to scattering theory. We also refer to (\cite{Ashida, FMW1, FMW2, FMW3, Hi} and the references therein) for some recent works on the widths of resonances for systems of coupled Schr\"odinger operators for different potentials and at different energy levels.

\section{Statement of the results}\label{Section2} 

Let $\mathcal{H}_N$ denotes the space of $N\times N$ hermitian matrices endowed with the norm $\Vert \cdot\Vert_{N\times N}$, where for $\mathcal{M}\in \mathcal{H}_N$, 
\begin{equation}\label{matrixnorm}
\Vert \mathcal{M} \Vert_{N\times N}:=\sup_{\{w\in \mathbb R^N; |w|\leq 1\}}|\mathcal{M} w|.
\end{equation}

\subsection{Resolvent estimates}  Consider the semiclassical Schr\"odinger operator on $L^2(\mathbb R^d;\mathbb C^N)$, $d\neq 2$,
$$
P(h):= -h^2\Delta \cdot I_N + V(x),
$$
where $I_N$ is the identity $N\times 	N$ matrix and $h>0$ is the semiclassical parameter. Using the polar coordinates $\mathbb R^d \ni x=(r,\omega)\in \mathbb R_+\times \mathbb S^{d-1}$, where $ \mathbb S^{d-1}$ denotes the unit sphere on $\mathbb R^d$, we assume the following conditions on the potential $V$.

\noindent
\textbf{Assumption (A1).} $V: \mathbb R_+\times \mathbb S^{d-1} \rightarrow \mathcal{H}_N$ and its distributional derivative $\partial_r V$ are bounded on $\mathbb R_+\times \mathbb S^{d-1}$, i.e., 
\begin{equation}\label{hyp1}
V,\partial_r V\in   L^{\infty}(\mathbb R_+\times \mathbb S^{d-1};\mathcal{H}_N),
\end{equation}
and has the following long-range behavior at infinity: 

\noindent
\textbf{Assumption (A2).} There exist a constant hermitian matrix $V_{\infty}\in \mathcal{H}_N$ and $\rho_0>0$ such that
\begin{equation}\label{hyp2}
\Vert V(r,\omega) - V_{\infty} \Vert_{N\times N} \leq  (1+r)^{-\rho_0} , \quad \Vert \partial_r V(r,\omega)  \Vert_{N\times N} \leq  (1+r)^{-\rho_0-1},
\end{equation}
for all $(r,\omega)\in \mathbb R_+\times \mathbb S^{d-1}$.

We have the following resolvent estimates.

\begin{theorem}\label{Exponential estimate resolvent} \normalfont
Assume \textbf{(A1)} and \textbf{(A2)}. For any $E>\Vert V_{\infty} \Vert_{N\times N}$ and $s>\frac{1}{2}$, there exist $C,R_0,h_0>0$ such that for all $h\in (0,h_0]$ and $\varepsilon>0$, the following estimates hold
\begin{equation}\label{estimation exponentielle du th}
 \Vert R_h(E\pm i \varepsilon)  \Vert_{L^{2,s}(\mathbb R^d;\mathbb C^N)\rightarrow L^{2,-s}(\mathbb R^d;\mathbb C^N)}   \leq e^{\frac{C}{h}},
\end{equation}
\begin{equation}\label{estimation du th}
\big\Vert   \mathbbm 1_{\{x\in \mathbb R^d;\vert x\vert\geq R_0\}}R_h(E\pm i \varepsilon) \mathbbm 1_{\{x\in \mathbb R^d;\vert x\vert\geq R_0\}}   \big\Vert_{L^{2,s}(\mathbb R^d;\mathbb C^N)\rightarrow L^{2,-s}(\mathbb R^d;\mathbb C^N)}  \leq C h^{-1}.
\end{equation}
Here $\mathbbm 1_{\{x\in \mathbb R^d;\vert x\vert\geq R_0\}}$ is the characteristic function of the set $\{x\in \mathbb R^d;\vert x\vert\geq R_0\}$.
\end{theorem}

\begin{remark}\normalfont
Estimates on resolvent truncated outside a large compact set are important in scattering theory. Indeed, the operator $\mathbbm 1_{\{x\in \mathbb R^d;\vert x\vert\geq R_0\}}R_h(z) \mathbbm 1_{\{x\in \mathbb R^d;\vert x\vert\geq R_0\}}$ appears for instance in the representation of the scattering amplitude for compactly supported perturbations (see \cite{Petkov}). In \cite{Assal1}, we apply \eqref{estimation du th} to prove estimates on the scattering amplitude for Schr\"odinger operators with matrix-valued potentials.

\end{remark}
\subsection{Resonance free domains} Now, we state our results on the resonances of $P(h)$. To define the resonances of $P(h)$, we need the following assumption on the potential $V$. 

\noindent
\textbf{Assumption ($\textbf{Hol}_{\infty}$).} $V\in C^{\infty}(\mathbb R^d;\mathcal{H}_N)$ and extends to an analytic function on $\mathcal{S} \subset \mathbb C^d$
\begin{equation}\label{hyp holomorph}
\mathcal{S} := \{x\in \mathbb C^d; \, |{\rm Im}\, x|\leq c_0\, \langle {\rm Re}\, x\rangle, \; |{\rm Re}\, x|>\kappa\},
\end{equation}
for some constants $c_0, \kappa>0$. Moreover, there exist $\rho_0>0$ and a constant $C>0$ such that for all $x\in  \mathcal{S}$ 
\begin{equation}\label{decroissance a l'infini}
\Vert V(x)-V_{\infty}\Vert_{N\times N} \leq C \langle x \rangle^{-\rho_0}, \quad \langle x \rangle:= (1+ \vert x \vert^2)^{\frac12}.
\end{equation}

Under this assumption, we can define the resonances of $P(h)$ near the real axis by the method of complex distortion as it was done in \cite{Hunziker, Nedelec} (see also \cite{Aguilar, Sj-Zw} for an alternative approach). Let $A\gg 1$ be a large constant, and let $F:\mathbb R^d \rightarrow \mathbb R^d$ be a smooth vector-field such that 
\begin{equation}\label{vector field}
F(x)= \left\{ \begin{array}{lll}
0 \;\; {\rm for} \; |x|\leq A \\
x    \;\; {\rm for} \; |x|\geq A+1.
\end{array}\right.
\end{equation}
We introduce the one-parameter family of unitary distortion 
$$
C_0^{\infty}(\mathbb R^d;\mathbb C^N) \ni f \longmapsto  U_{\omega} f(x) :=  \vert J_{\phi_{\omega}(x)}\vert^{\frac{1}{2}} f(\phi_{\omega}(x)) , \quad \omega \in \mathbb R,
$$
where $\phi_{\omega}(x):=x+\omega F(x)$ and $J_{\phi_{\omega}(x)}:= \det(1+\omega \nabla F(x))$ is the Jacobian of $\phi_{\omega}(x)$.

For $\omega\in \mathbb R$ small enough, $U_{\omega}$ extends to a unitary operator on $L^2(\mathbb R^d;\mathbb C^N)$. We define
\begin{equation}\label{distorted operator}
P_{\omega}(h) := U_{\omega} P(h) (U_{\omega})^{-1} .
\end{equation}
Under the assumption \textbf{($\textbf{Hol}_{\infty}$)}  on the potential $V$, the operator $P_{\omega}(h)$ is a differential operator with analytic coefficients with respect to $\omega$, and can therefore be continued in a unique way to small enough complex values of $\omega$. Therefore, the distorted operator  
\begin{equation}\label{Distorted operator}
P_{\theta}(h):=U_{i\theta} P(h) (U_{i\theta})^{-1}
\end{equation}
is well defined for $\theta>0$ small enough. By Weyl perturbation theorem its essential spectrum is given by 
$$
\sigma_{\text{ess}}(P_{\theta}(h)) = \bigcup_{j=1}^{N} (\lambda_{j} + e^{-2i\theta} \mathbb R_+).
$$
Hence, the spectrum of $P_{\theta}(h)$ in the complex sector
$$
S_{\theta} := \left(\lambda_{1}+ e^{-2i[0,\theta)}\mathbb R^*_+\right) \setminus \left(\cup_{j=1}^N (\lambda_{j}+e^{-2i\theta}\mathbb R_+) \right)
$$
is discrete, consists on isolated eigenvalues with finite multiplicities. Moreover, standard arguments (see e.g. \cite{Helffer}) show that $\sigma(P_{\theta}(h))\cap S_{\theta}$ does not depends on the particular choice of the vector-field $F$ and for any $0<\theta<\theta'$ small enough and $h>0$ fixed, we have 
$$
\sigma(P_{\theta}(h)) \cap S_{\theta} = \sigma(P_{\theta'}(h)) \cap S_{\theta}.
$$

The resonances of $P(h)$ in $S_{\theta}$ are defined as the eigenvalues of $P_{\theta}(h)$ in $S_{\theta}$, or equivalently as the eigenvalues of $P_{\theta'}(h)$ in $S_{\theta}$ for all $0<\theta<\theta'$ small enough. In the following, we denote ${\rm Res}\, (P(h)) $ the set of resonances of $P(h)$. 

The estimate \eqref{estimation exponentielle du th} entails the following result on the absence of resonances in an exponentially small band below $(\Vert V_{\infty} \Vert_{N\times N}, +\infty)$. 

\begin{theorem}\label{exponential resonance free region} \normalfont
Assume \textbf{($\textbf{Hol}_{\infty}$)}. For any compact interval $J \subset (\Vert V_{\infty}\Vert_{N\times N},+\infty)$, there exist a constant $C>0$ and $h_0\in (0,1]$ such that for all $h\in (0,h_0]$, we have
$$
{\rm Res} (P(h)) \cap \left\{ z\in \mathbb C;\, {\rm Re}\, z \in J \,\, {\rm and} \,\, {\rm Im}\, z\geq -Ce^{-C/h} \right\} = \emptyset.
$$
\end{theorem}

\begin{remark} \normalfont
Notice that assumption \textbf{($\textbf{Hol}_{\infty}$)} and the Cauchy formula imply assumptions \textbf{(A1)} and \textbf{(A2)}.
\end{remark}

In our proof of Theorem \ref{exponential resonance free region}, we shall use the following result which generalizes the well-known result of \cite{Martinez1} to the case of Schr\"odinger operators with matrix-valued potentials. 

\begin{theorem}\label{region logarithmique} \normalfont
Assume \textbf{($\textbf{Hol}_{\infty}$)}. Let $E_0>\Vert V_{\infty}\Vert_{N\times N}$ and suppose that there exists $G\in C^{\infty}(T^*\mathbb R^d;\mathbb R)$ such that \eqref{la fonction fuite} holds on $\Sigma_{E_0}$. Then there exists $\varepsilon_0>0$ such that for all $C>0$, there exists $h_C\in (0,1]$ such that for all $0<h\leq h_C$, we have
$$
{\rm Res} (P(h)) \cap  \bigg\{ z\in \mathbb C;\, {\rm Re}\, z \in [E_0-\varepsilon_0,E_0+\varepsilon_0] \,\, {\rm and} \,\, {\rm Im}\, z\geq -C h \vert \ln h \vert  \bigg\} = \emptyset.
$$
\end{theorem}

\begin{remark} \normalfont
In the scalar case $N=1$, a quantitative version of the previous result in terms of an estimate on the distorted resolvent was proved in \cite{Nakamura1} (see also \cite{Sj-Zw1}). More precisely, if $P_{\theta}(h)$ denotes the distorted operator with $\theta=\theta(h)=C' h \vert \ln h\vert$, $C'\gg C$, then for any $C>0$ there exist a constant $C''>0$ and $h_C>0$ such that 
\begin{equation}\label{quantitative estimate}
\Vert (P_{\theta}(h)-z)^{-1}\Vert \leq \frac{C''}{h} \exp(C''\vert {\rm Im}\,z \vert/h), 
\end{equation}
uniformly for $z\in  \{ z\in \mathbb C;\, {\rm Re}\, z \in [E_0-\varepsilon_0,E_0+\varepsilon_0] \,\, {\rm and} \,\, {\rm Im}\, z\geq -C h \vert \ln h \vert  \}$ and $h\in (0,h_C]$. 

In Lemma \ref{Lemma non-captif}, we prove a weaker estimate (see \eqref{distortion non cap}) on the distorted resolvent which coincides with \eqref{quantitative estimate} for $ {\rm Im}\, z = C h \vert \ln h \vert$. As in \cite{Nakamura1}, using a version of the non-trapping estimate and  the semiclassical maximum principle, we can prove the same estimate \eqref{quantitative estimate} for our matrix-valued operator. This will be proved in details in \cite{Assal1}.
\end{remark}

\section{Resolvent estimates}

In this section, we present the main ideas of the proof of Theorem \ref{Exponential estimate resolvent} referring to \cite{Assal1} for the details. We follow the approach developed in \cite{Datchev} which we adapt in our context of matrix-valued operator. The main step in the proof of estimates \eqref{estimation exponentielle du th} and \eqref{estimation du th} is the following global weighted Carleman estimate. In the following $s>\frac{1}{2}$ and $E> \Vert V_{\infty} \Vert_{N\times N}$ are fixed.

\begin{proposition}\label{prop Carleman}
There exist $R_0,h_0,C>0$ and a positive radial function $\varphi=\varphi(r)\in C^{\infty}(\mathbb R_+;\mathbb R_+)$ with $\varphi'\geq 0$ and ${\rm supp}\, \varphi' = [0,R_0]$ such that
the following estimate holds
\begin{equation}\label{Global Carleman estimate}
h^2\big\Vert  e^{\varphi/h} v \big\Vert^2_{L^{2,-s}(\mathbb R^d;\mathbb C^N)} \leq C  \big\Vert  e^{\varphi/h} \big(P(h)-(E+ i\varepsilon) I_N \big)v \big\Vert^2_{L^{2,s}(\mathbb R^d;\mathbb C^N)} + C\varepsilon h\big\Vert e^{\varphi/h}v \big\Vert^2_{L^2(\mathbb R^d;\mathbb C^N)},
\end{equation} 
for all $\varepsilon\geq 0$, $h\in (0,h_0]$ and $v\in C_0^{\infty}(\mathbb R^d;\mathbb C^N)$. 
\end{proposition}

The proof of this result relies mainly on two steps.

\noindent
\textbf{Step 1.} The first step consists in the construction of a weight function $\varphi\in C^{\infty}(\mathbb R_+;\mathbb R_+)$ such that for some $0<R<R_0$, $\varphi$ increases linearly on $[0,R]$, $\varphi$ increases slowly on $[R,R_0]$ and $\varphi={\rm Cste}$ on $[R_0,+\infty)$ (see Figure \ref{Fig}), so that it satisfies a nice estimate in relation with the potential $V$, more precisely,
\begin{equation}\label{main estimate weight function}
\bigg(\frac{1}{m'(r)}\partial_r \big[ m(r) \big(E I_N- V_{\varphi}(r,\omega;h)\big) \big] w,w  \bigg)_{\mathbb C^N} \geq C \vert w \vert^2_{\mathbb C^N}
\end{equation}
for some constant $C>0$, uniformly with respect to $(r,\omega)\in (0,+\infty)\times \mathbb S^{d-1}$, $h\in (0,h_0]$ and $w\in \mathbb C^N$, where $m(r) := 1 - (1+r)^{1-2s}$ and 
\begin{equation}\label{effective potential}
V_{\varphi}(r,\omega;h) := V(r,\omega) - \big( (\varphi'(r))^2 - h \varphi''(r)  \big) I_N.
\end{equation}
Here $(\cdot,\cdot)_{\mathbb C^N}$ and $\vert \cdot \vert_{\mathbb C^N}$ denote the hermitian inner product and norm in $\mathbb C^N$.

\begin{center}
\begin{figure}[h]
\begin{tikzpicture}[scale=0.7][line cap=round,line join=round]
\clip(-5,-1) rectangle (9.96,3.76);
\draw [shift={(4,-2)},line width=1.5pt]  plot[domain=1.5731659906058693:2.356194490192345,variable=\t]({1*4.2200118483246*cos(\t r)+0*4.2200118483246*sin(\t r)},{0*4.2200118483246*cos(\t r)+1*4.2200118483246*sin(\t r)});
\draw [line width=1.5pt] (0,0)-- (1.016001005362099,0.9839989946379015);
\draw [line width=1.5pt] (3.99,2.22)-- (8.99,2.26);
\draw [line width=1pt,dotted] (0,2.22)-- (4,2.22);
\draw [line width=1pt,dotted] (1.2692226702311507,1.2173521994367547)-- (1.2692226702311507,0.06);
\draw [line width=1pt,dotted] (3.99,2.22)-- (4,0);
\draw (0.000000001,-0.15) node[anchor=north west] {$0$};
\draw (0.85,-0.15) node[anchor=north west] {$R$};
\draw (3.72,-0.15) node[anchor=north west] {$R_0$};
\draw [line width=1pt,dotted] (1.2692226702311507,1.2173521994367547)-- (0,1.22);
\draw (8.7,-0.29) node{$r$};
\draw (-1.4,4) node[anchor=north west] {$\varphi(r)$};
\draw (1.2,2.1) node[anchor=north west] {$\varphi$};
\draw (-1.7,2.73) node[anchor=north west] {$\varphi(R_0)$};
\draw [very thick,->] (0,-1) -- (0,3.7);
\draw[very thick,->] (-1,0) -- (8.9,0) coordinate;
\end{tikzpicture}
\caption{The weight function $\varphi$} \label{Fig}
\end{figure}
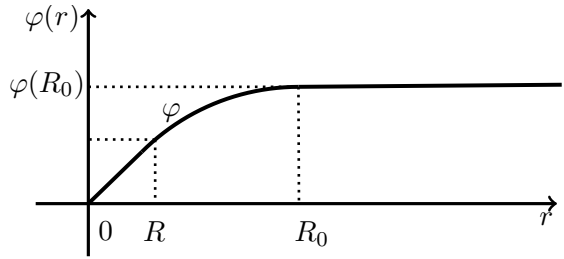
\end{center}

\noindent 
\textbf{Step 2.} Now, using the weighted function $\varphi$ constructed above, we introduce the conjugated operator 
\begin{align*}
\mathcal{P}_{\varphi} = \mathcal{P}_{\varphi}(h) &:=  e^{\varphi/h} r^{(d-1)/2} \big(P(h)- (E+i\varepsilon)I_N \big) r^{-(d-1)/2}  e^{-\varphi/h}.
 \end{align*}
We have 
 $$
\mathcal{P}_{\varphi} = \big(- h^2 \partial_r^2 + 2h \varphi'(r) \partial_r + \mathcal{Q} \big) \cdot I_N+ V_{\varphi}(r,\omega;h) - (E+i\varepsilon) I_N,
$$
where $V_\varphi$ is the effective potential defined by \eqref{effective potential} and 
$$
\mathcal{Q}= \mathcal{Q}(h):= \left\{ \begin{array}{lll} 
0 &(d=1),&\\
h^2 r^{-2} \big(- \Delta_{\mathbb S^{d-1}} + \frac{(d-1)(d-3)}{4} \big) &(d\geq 3).&
\end{array}\right.
$$
Here $\Delta_{\mathbb S^{d-1}}$ denotes the Laplacian on the unit sphere $\mathbb S^{d-1}$.

Passing to polar coordinates using that $L^2(\mathbb R^d,dx) = L^2(\mathbb R_+\times \mathbb S^{d-1}, r^{d-1} drd\omega)$, the estimate \eqref{Global Carleman estimate} is equivalent to the following one 
\begin{equation}\label{equivalence de l'estimation}
\iint_{r,\omega} m'(r) \vert u(r,\omega)\vert_{\mathbb C^N}^2 dr d\omega \leq \frac{C}{h^2}   \iint_{r,\omega}   \frac{\vert \mathcal{P}_{\varphi} u(r,\omega)\vert_{\mathbb C^N}^2}{m'(r)} dr d\omega+ \frac{C \varepsilon}{h}  \iint_{r,\omega} \vert u(r,\omega)\vert_{\mathbb C^N}^2 dr d\omega,
\end{equation}
for all $u\in e^{\varphi/h} r^{(d-1)/2} C_0^{\infty}(\mathbb R^d;\mathbb C^N)$. 

Let $\langle \cdot, \cdot \rangle_{\mathbb S^{d-1}}$ and $ \Vert \cdot \Vert_{\mathbb S^{d-1}}$ denote the inner product and norm in $L^2(\mathbb S^{d-1};\mathbb C^N)$, and we introduce the functional
$$
\mathcal{L}_h(r) := \Vert h  u'(r,\omega) \Vert_{\mathbb S^{d-1}}^2 - \big\langle (\mathcal{Q} \cdot I_N+ V_{\varphi}(r,\omega;h) - E I_N) u, u \big\rangle_{\mathbb S^{d-1}}, \quad r>0,
$$
where here and in the sequel prime notation always denote differentiation with respect to $r$, for instance $u':=\partial_ru$. Using the self-adjointness of 
$$
\mathcal{Q} \cdot  I_N+ V_{\varphi} - EI_N = \mathcal{P}_{\varphi} - \big( -h^2 \partial_r^2 + 2h \varphi'(r) \partial_r - i \varepsilon \big) \cdot I_N,
$$ 
we get 
\begin{align*}
\mathcal{L}_h'(r)=&-2 \, \text{Re}\, \langle \mathcal{P}_{\varphi} u, u' \rangle_{\mathbb S^{d-1}} + 4 h\varphi'(r)\Vert u'\Vert_{\mathbb S^{d-1}}^2 + 2 \varepsilon \; \text{Im}\, \langle u,u'\rangle_{\mathbb S^{d-1}} \\
&+ 2 \, r^{-1}  \langle \mathcal{Q} \cdot I_N u, u \rangle_{\mathbb S^{d-1}} - \langle V'_{\varphi} u, u \rangle_{\mathbb S^{d-1}}.
\end{align*}
It follows that
\begin{align*}
(m \mathcal{L}_h)'(r) = &-2 m \,\text{Re}\, \langle \mathcal{P}_{\varphi} u, u' \rangle_{\mathbb S^{d-1}}  + (4 h^{-1}m \varphi' + m') \Vert h u' \Vert_{\mathbb S^{d-1}} ^2 + 2 m \varepsilon \;   \text{Im}\, \langle u,u'\rangle_{\mathbb S^{d-1}}   \\
&+ (2 m r^{-1}-m') \langle \mathcal{Q} \cdot I_N u, u \rangle_{\mathbb S^{d-1}}
   + \big\langle \partial_r\big[ m(EI_N-V_{\varphi})\big] u,u  \big\rangle_{\mathbb S^{d-1}}. 
\end{align*}
First, using the fact that 
$$
m\varphi' \geq 0, \;\; m' >0, \;\; \mathcal{Q} \geq 0, \;\; 2 mr^{-1} - m' >0,
$$
we obtain 
\begin{align*}
(m\mathcal{L}_h)'(r) \geq  &-2 m \text{Re}\, \langle \mathcal{P}_{\varphi} u, u' \rangle_{\mathbb S^{d-1}}   +  m' \Vert h u' \Vert_{\mathbb S^{d-1}} ^2  + \langle \partial_r\big[ m(EI_N-V_{\varphi})\big]  u,u \rangle_{\mathbb S^{d-1}} \\
&+ 2 m \varepsilon \;   \text{Im}\, \langle u,u'\rangle_{\mathbb S^{d-1}}.
\end{align*}
Then, using the inequality $-2\, \text{Re}\, \langle a,b \rangle + \Vert b\Vert^2 \geq - \Vert a\Vert^2$, we get
\begin{align*}
(m\mathcal{L}_h)'(r) &\geq  - \frac{1}{h^2 m'} \Vert \mathcal{P}_{\varphi}u \Vert^2_{\mathbb S^{d-1}}  + 2 m \varepsilon \;  \text{Im}\, \langle u,u'\rangle_{\mathbb S^{d-1}}  + \langle  \partial_r\big[ m(EI_N-V_{\varphi})\big] u,u \rangle_{\mathbb S^{d-1}}.
\end{align*}
By integrating with respect to $r$ using the fact that $\int_0^{+\infty} (m \mathcal{L}_h)'(r) dr  = 0$ and the Cauchy-Schwarz inequality, we obtain 
$$
\iint_{r,\omega} \big(   \partial_r\big[ m(EI_N-V_{\varphi})\big] u, u \big)_{\mathbb C^N} dr d\omega \leq \frac{1}{h^2} \iint_{r,\omega} \frac{\vert \mathcal{P}_{\varphi} u \vert^2_{\mathbb C^N}}{m'(r)} dr d\omega+ 2 \varepsilon \iint_{r,\omega}  \vert u \vert_{\mathbb C^N}  \vert u' \vert_{\mathbb C^N} dr d\omega .
$$
It follows from estimate \eqref{main estimate weight function} that 
\begin{equation}\label{Latst}
C \iint_{r,\omega} m'(r)\vert w\vert^2_{\mathbb C^N} dr d\omega \leq \frac{1}{h^2} \iint_{r,\omega} \frac{\vert \mathcal{P}_{\varphi} u \vert^2_{\mathbb C^N}}{m'(r)} dr d\omega+ 2 \varepsilon \iint_{r,\omega}  \vert u \vert_{\mathbb C^N}  \vert u' \vert_{\mathbb C^N} dr d\omega,
\end{equation}
for some constant $C>0$. On the other hand, it is not difficult to prove that there exists a constant $C'>0$ such that 
\begin{equation}\label{last equation proof}
 2 \varepsilon \iint_{r,\omega}  \vert u \vert_{\mathbb C^N}  \vert u' \vert_{\mathbb C^N} dr d\omega  \leq \frac{C' \varepsilon}{h} \iint_{r,\omega} \vert u \vert_{\mathbb C^N}^2 dr d\omega + \frac{C' \varepsilon}{h}  \iint_{r,\omega} \vert \mathcal{P}_{\varphi} u \vert_{\mathbb C^N}^2 dr d\omega.
\end{equation}
Putting together \eqref{Latst} and \eqref{last equation proof} we obtain \eqref{equivalence de l'estimation}. This ends the proof of Proposition \ref{prop Carleman}.
\begin{flushright}
$\square$
\end{flushright}

Now, using the fact that $\varphi(r)= \frac{C_0}{2}$ for $r \geq R_0$, with $C_0=2 \max_{(0,+\infty)} \varphi$, and estimate \eqref{Global Carleman estimate}, we obtain 
\begin{align}\label{aw1}
e^{-C_0/h} \Vert   \mathbbm 1_{\{x\in \mathbb R^d;\vert x\vert\leq R_0\}}v \Vert^2_{L^{2,-s}}  &+    \Vert   \mathbbm 1_{\{x\in \mathbb R^d;\vert x\vert\geq R_0\}} v \Vert^2_{L^{2,-s}} \leq
\nonumber \\
&\frac{C}{h^2}  \big\Vert \big(P-(E+i\varepsilon)I_N \big)v \big\Vert^2_{L^{2,s}} + \frac{C\varepsilon}{h} \big\Vert  v \big\Vert^2_{L^2},
\end{align}
for some constant $C>0$, uniformly for $v\in C_0^{\infty}(\mathbb R^d;\mathbb C^N)$, $\varepsilon\geq 0$ and $h>0$ small enough. On the other hand, using the selfadjointness of $P$ and the Cauchy-Schwarz inequality, we get, for any $\varepsilon\geq 0$, $C'>0$ and $h>0$ small enough,
\begin{align}
2 \varepsilon \Vert v\Vert^2_{L^2} &=  - 2 \, {\rm Im}\, \langle \big(P-(E+i\varepsilon)I_N \big) v, v \rangle_{L^2} \nonumber \\ 
&\leq  \frac{C'}{h}  \Vert  \mathbbm 1_{\{x\in \mathbb R^d;\vert x\vert\geq R_0\}}\big(P-(E+i\varepsilon)I_N \big) v \Vert^2_{L^{2,s}} + \frac{h}{C'} \Vert  \mathbbm 1_{\{x\in \mathbb R^d;\vert x\vert\geq R_0\}} v \Vert^2_{L^{2,-s}}  \nonumber \\
&+  e^{2C_0/h}  \Vert  \mathbbm 1_{\{x\in \mathbb R^d;\vert x\vert\leq R_0\}}\big(P-(E+i\varepsilon)I_N \big) v \Vert^2_{L^{2,s}} + e^{-2C_0/h} \Vert\mathbbm 1_{\{x\in \mathbb R^d;\vert x\vert\leq R_0\}} v \Vert^2_{L^{2,-s}}. \label{aw2}
\end{align}
Combining \eqref{aw1} and \eqref{aw2}, we get 
\begin{align*} 
e^{-C/h} &\Vert  \mathbbm 1_{\{x\in \mathbb R^d;\vert x\vert\leq R_0\}}v \Vert^2_{L^{2,-s}} +    \Vert   \mathbbm 1_{\{x\in \mathbb R^d;\vert x\vert\geq R_0\}} v \Vert^2_{L^{2,-s}}  \leq \nonumber \\
& e^{C/h}  \big\Vert  \mathbbm 1_{\{x\in \mathbb R^d\vert x\vert\leq R_0\}} \big(P-(E+i\varepsilon)I_N \big)v \big\Vert^2_{L^{2,s}} + \frac{C}{h^2} \big\Vert  \mathbbm 1_{\{x\in \mathbb R^d;\vert x\vert\geq R_0\}} \big(P-(E+i\varepsilon)I_N \big)v \big\Vert^2_{L^{2,s}},
\end{align*}
uniformly for $v\in C_0^{\infty}(\mathbb R^d;\mathbb C^N)$, $\varepsilon\geq 0$ and $h>0$ small enough.

Using this estimate, the proof can be finished by the density argument of \cite{Datchev}. 
\begin{flushright}
$\square$
\end{flushright}

\section{Resonance free domains}

This section is devoted to the proofs of Theorems \ref{exponential resonance free region} and \ref{region logarithmique}. We use the ordinary notations and some basic results of semiclassical analysis referring for example to the textbooks \cite{Dimassi,Zworski} for a clear presentation of this theory.

We first prove Theorem \ref{region logarithmique} using the approach to conjugate operators developed in \cite{Sj-Zw1}. 

\subsection{Proof of Theorem \ref{region logarithmique}}

Fix $E_0>\Vert V_{\infty}\Vert_{N\times N}$, and let 
$$
p(x,\xi):=\vert \xi \vert^2 I_N + V(x), \quad (x,\xi)\in T^*\mathbb R^{d},
$$
be the semiclassical symbol of $P(h)$. Let $G\in C^{\infty}(T^*\mathbb R^d;\mathbb R)$ be an escape function associated with $p$ at $E_0$, i.e., $G$ satisfies \eqref{la fonction fuite} on $\Sigma_{E_0}$. For  $|(x,\xi)|$ large enough, the function $T^*\mathbb R^{d}\ni (x,\xi)\mapsto x\cdot \xi$ is an escape function associated with $p$ at $E_0$. Indeed, we have 
$$
\{p,x\cdot \xi\}(x,\xi)= 2 \vert \xi \vert^2 I_N - x \cdot \nabla_x V(x), \quad (x,\xi)\in T^*\mathbb R^{d}.
$$ 
The assumption \textbf{($\textbf{Hol}_{\infty}$)} and the Cauchy formula imply that $x \cdot \nabla_x V(x) \rightarrow 0$ as $\vert x\vert\rightarrow + \infty$. On the other hand, by the assumption \textbf{($\textbf{Hol}_{\infty}$)}, $\vert \xi \vert^2\geq (E_0-\Vert V_{\infty}\Vert_{N\times N})/2$ on the energy surface $\Sigma_{E_0}$ for $|(x,\xi)|$ large enough. Therefore, 
$$
\{p,x\cdot \xi\}(x,\xi)\geq  (E_0-\Vert V_{\infty}\Vert_{N\times N})/2,
$$
for all $(x,\xi)\in \Sigma_{E_0}\cap \{|(x,\xi)|\gg 1\}$. Thus, without any loss of generality, we may assume that $G(x,\xi)=x\cdot \xi$ for $|(x,\xi)|$ large enough. 

We set
$$
\mathcal{G}(x,\xi) := G(x,\xi) - F(x)\cdot \xi \in C_0^{\infty}(\mathbb R^d;\mathbb R),
$$
where $F$ is the vector field used in the complex distortion (see \eqref{vector field}). Let $\mathcal{G}_h:= {\rm Op}^w_h(\mathcal{G})$ be the pseudodifferential operator with symbol $\mathcal{G}$. 

Let $M>0$ be independent of $h$ and set $\theta_1=\theta_1(h)=i M \kappa(h)$ with $\kappa(h):= h \vert \ln h \vert$. We introduce the conjugated operator 
$$
\widetilde{P}_{\theta_1}(h) := e^{-\frac{M \kappa(h)}{h}\mathcal{G}_h} P_{\theta_1}(h) e^{\frac{M \kappa(h)}{h}\mathcal{G}_h}.
$$
Since $\mathcal{G}$ is compactly supported it follows by the Calder\'on-Vaillancourt theorem (see for instance \cite[Chapter 7]{Dimassi}) that $\mathcal{G}_h$ is bounded in $L^2(\mathbb R^d;\mathbb C^N)$ and then the operators  $e^{\pm \frac{M \kappa(h)}{h}\mathcal{G}_h}$ are well defined.

Let $\varepsilon_0>0$ be small enough such that \eqref{la fonction fuite} holds on $\Sigma_{I_{\varepsilon_0}}:=\bigcup_{E\in I_{\varepsilon_0}}\Sigma_{E}$, i.e., there exists $C>0$ such that
\begin{equation}\label{vois}
\{p,G\}(x,\xi)\geq C , \quad \forall (x,\xi)\in \Sigma_{I_{\varepsilon_0}},
\end{equation}
in the sense of hermitian matrices, where $I_{\varepsilon_0}:=[E_0-\varepsilon_0,E_0 +\varepsilon_0]$. For $\eta>0$, we introduce the complex region
$$
\Gamma_{\eta}:=I_{\varepsilon_0} -i [0,\eta \kappa(h)].
$$
Our objective is to prove the following Lemma from which Theorem \ref{region logarithmique} follows.
\begin{lemma}\label{Lemma non-captif}
There exists a constant $c>0$ such that for all $M>0$, there exists $h_M\in (0,1]$ such that the operator $\widetilde{P}_{\theta_1}(h) - z$ is invertible for every $z\in \Gamma_{cM}$ and $h\in (0,h_M]$, and we have 
$$
\Vert (\widetilde{P}_{\theta_1}(h) - z)^{-1} \Vert = \mathcal{O}(\kappa(h)^{-1}),
$$
uniformly for $z\in \Gamma_{cM}$ and $h\in (0,h_M]$. 
\end{lemma}
\begin{remark}
In particular, from the above Lemma we get immediately that the operator $P_{\theta_1}(h) - z$ is invertible for every $z\in \Gamma_{cM}$ and $h\in (0,h_M]$, hence $P(h)$ has no resonances in $\Gamma_{cM}$ for all $M>0$ and $h\in (0,h_M]$. Furthermore, we have the following estimate on the distorted resolvent
\begin{equation}\label{distortion non cap}
\Vert (P_{\theta_1}(h) - z)^{-1} \Vert = \mathcal{O}(h^{-C}),
\end{equation}
uniformly for $z\in \Gamma_{cM}$ and $h\in (0,h_M]$, for some constant $C>0$.
\end{remark}

\begin{proof}
We have 
$$
\widetilde{P}_{\theta_1}(h) =  e^{-\frac{M \kappa(h)}{h}{ad}_{\mathcal{G}_h}} P_{\theta_1}(h) \sim \sum_{k=0}^{+\infty} \frac{(-M \kappa(h))^k}{\fact{k}} \big( \frac{1}{h} {\rm ad}_{\mathcal{G}_h}  \big)^k P_{\theta_1}(h),
$$
where here we use the usual notation ${\rm ad}_{A}B:= [A,B]$ for the commutator. The fact that $\mathcal{G}$ is scalar-valued and compactly supported ensures that ${\rm ad}_{\mathcal{G}_h}P_{\theta_1}(h) = [\mathcal{G}_h,P_{\theta_1}(h)]=\mathcal{O}(h)$ (in norm $\mathcal{L}(L^2)$) and then the previous asymptotic expansion makes sense since $\kappa(h)\rightarrow 0$ as $h$ tends to $0$. In particular, we have 
$$
\widetilde{P}_{\theta_1}(h) = P_{\theta_1}(h)  - \frac{M \kappa(h)}{h}[\mathcal{G}_h,P_{\theta_1}(h)] + \mathcal{O}(M^2 \kappa(h)^2).
$$
Let $p_{\theta_1}, \widetilde{p}_{\theta_1}$ be the semiclassical symbols corresponding to $P_{\theta_1}(h)$ and $\widetilde{P}_{\theta_1}(h)$ respectively. By the $h$-pseudodifferential symbolic calculus (see for instance \cite{Dimassi,Zworski}), we have 
\begin{equation}\label{equ1}
\widetilde{p}_{\theta_1}(x,\xi) = p_{\theta_1}(x,\xi) - iM \kappa(h) \{\mathcal{G},p_{\theta_1}\}(x,\xi) + \mathcal{O}(M^2 \kappa(h)^2). 
\end{equation}
On the other hand, by Taylor's expansion of $p_{\theta_1}$ with respect to $\theta_1$, we get 
\begin{equation}\label{equ2}
p_{\theta_1}(x,\xi) = p(x,\xi) - i M \kappa(h) \{p, F(x)\cdot\xi\}(x,\xi) + \mathcal{O}(M^2 \kappa(h)^2).
\end{equation}
Combining \eqref{equ1} and \eqref{equ2}, we obtain 
\begin{equation}\label{part imaginaire}
{\rm Im}\, \widetilde{p}_{\theta_1}(x,\xi) = -M \kappa(h) \{p,\mathcal{G}+F(x)\cdot \xi\}(x,\xi) + \mathcal{O}(M^2 \kappa(h)^2),
\end{equation}
$$
{\rm Re}\,\widetilde{p}_{\theta_1}(x,\xi) = p(x,\xi) + \mathcal{O}(M \kappa(h)).
$$
According to \eqref{vois}, there exists $C>0$ such that  
\begin{equation}\label{estimate on the imaginary part}
- {\rm Im}\, \widetilde{p}_{\theta_1}(x,\xi) \geq C M \kappa(h) , \quad \forall\, (x,\xi)\in \Sigma_{I_{\varepsilon_0}}.
\end{equation}
We write $\widetilde{P}_{\theta_1}(h) -z = A_{\theta_1}(h) - {\rm Re}\, z + i (B_{\theta_1}(h)-{\rm Im}\, z)$, where $A_{\theta_1}(h)$ and $B_{\theta_1}(h)$ are the self-adjoint operators given by 
$$
A_{\theta_1}(h) := \frac{1}{2}\big(\widetilde{P}_{\theta_1}(h) + (\widetilde{P}_{\theta_1}(h))^*\big) , \quad \quad B_{\theta_1}(h):= \frac{1}{2i}\big(\widetilde{P}_{\theta_1}(h) - (\widetilde{P}_{\theta_1}(h))^*\big).
$$
Let $\psi_{1},\psi_{2} \in C^{\infty}(\mathbb R^{2d};\mathbb R)$ be such that, for $I\Subset I_{\varepsilon_0}$,
\begin{equation}
\left\{ \begin{array}{lll}
\psi_{1}^2+\psi_{2}^2=1 \;\; {\rm on} \; \mathbb R^{2d} \\
{\psi_{1}}= 1    \;\; {\rm on} \;\;  \Sigma_{I} \;\;\; {\rm and} \;\; \text{supp}(\psi_{1}) \subset  \Sigma_{I_{\varepsilon_0}}.
\end{array}\right.
\end{equation}
According to Lemma 3.2 in \cite{Sj-Zw1}, there exist two self-adjoint operators $\Psi_{1}$ and $\Psi_{2}$ with principal symbols respectively $\psi_{1}$ and $\psi_{2}$ such that
\begin{equation}\label{Ident}
(\Psi_{1})^2 + (\Psi_{2})^2={\rm  Id} + \mathcal{O}(h^{\infty}) \quad \text{in}\;\; \mathcal{L}(L^2(\mathbb R^d;\mathbb C^N)).
\end{equation}
We denote by the same letters the operators $\Psi_{i}:=\Psi_{i} I_N$, $i=1,2$. On the support of $\psi_{1}$, we see from (\ref{estimate on the imaginary part}) that the principal symbol of $-B_{\theta_1}(h)$ is bounded from below by $CM \kappa(h)$. Thus, by G\aa rding's inequality (see e.g. \cite{Dimassi, Zworski}), we have for all $u\in L^2(\mathbb R^d;\mathbb C^N)$
\begin{align}\label{Gard}
 \Vert (\widetilde{P}_{\theta_1}(h)-z) \Psi_{1} u \Vert \cdot \Vert \Psi_{1} u\Vert &\geq \langle (\widetilde{P}_{\theta_1}(h) -z)\Psi_{1} u, \Psi_{1} u   \rangle  \nonumber \\
 & \geq   \langle ({\rm Im}\, \widetilde{P}_{\theta_1}(h)-{\rm Im}\, z) \Psi_{1} u, \Psi_{1} u\rangle \nonumber \\
 &= \langle ({\rm Im}\,z-B_{\theta_1}(h)) \Psi_{1} u, \Psi_{1} u \rangle \nonumber \\
 & \geq ( {\rm Im}\, z+CM\kappa(h)-\mathcal O(h))\Vert \Psi_{1} u \Vert^2 \nonumber \\
 & \geq  \frac{C}{3} M\kappa(h)\Vert \Psi_{1} u \Vert^2,
\end{align}
uniformly for $\Im z>-\frac{C}{3}M\kappa(h)$.


On the other hand, since $A_{\theta_1}(h)-{\rm Re}\, z$ is uniformly elliptic on the support of $\psi_{2}$ and ${\rm Re}\, z\in I_{\varepsilon_0}$, the symbolic calculus permits us to construct a parametrix  $R\in \mathcal{S}^0(\langle \xi\rangle^{-2})$ of $A_{\theta_1}(h)-{\rm Re}\, z$ such that, in the sense of corresponding symbols,
$$
R\# (A_{\theta_1}(h)-{\rm Re}\, z)\psi_{2}=\psi_{2}+{\mathcal O}(h^\infty),
$$
where $\#$ stands for the Weyl composition of symbols.  As a consequence, we obtain
\begin{eqnarray}\label{ell}
\Vert( \widetilde{P}_{\theta_1}(h)-z) \Psi_{2} u \Vert \geq \frac{1}{C'} \Vert \Psi_{2} u \Vert-{\mathcal O}(h^\infty)\Vert u\Vert^2.
\end{eqnarray}
for all $u\in L^2(\mathbb R^d;\mathbb C^N)$. Furthermore,  by means of standard  elliptic arguments, one can easily prove  the following semiclassical inequality, for $i=1,2$,
\begin{equation}\label{oubli6}
\Vert [\widetilde{P}_{\theta_1}(h), \Psi_{i}] u \Vert \leq C_2h (\Vert \widetilde{P}_{\theta_1}(h) u \Vert+\Vert u\Vert), \quad \forall\, u\in H^2({\mathbb R}^d; \mathbb C^N).
\end{equation}
Combining \eqref{Ident},  \eqref{Gard},  (\ref{ell}), and \eqref{oubli6}   with the estimate
\begin{equation}\label{Oubli1234}
\Vert (\widetilde{P}_{\theta_1}(h)-z)u \Vert^2 {=} \sum_{i=1}^2\Vert \Psi_{i}(\widetilde{P}_{\theta_1}(h)-z) u \Vert^2    -  {\mathcal   O}( h^{\infty})      \Vert (\widetilde{P}_{\theta_1}(h)-z) u \Vert^2 
\end{equation}
$$
\geq \frac{1}{2}\sum_{i=1}^2 \Vert (\widetilde{P}_{\theta_1}(h)- z)\Psi_{i} u \Vert^2-\sum_{i=1}^2\Vert [\widetilde{P}_{\theta_1}(h), \Psi_{i}] u \Vert^2 -{\mathcal  O}( h^{\infty})   \Vert (\widetilde{P}_{\theta_1}(h)-z) u \Vert^2 ,
$$
we deduce, for  $z\in \Gamma_{cM}:= \{z\in \mathbb C;\; {\rm Re}\, z\in I_{\varepsilon_0} \; {\rm and}\; {\rm Im}\, z \geq -cM \kappa(h)\}$ (with $c>0$ independent of $M$ and $h$) and sufficiently small $h$,
\begin{equation}\label{Est}
\Vert (\widetilde{P}_{\theta_1}(h)-z)  u\Vert \geq \frac{\kappa(h)}{C} \Vert u\Vert.
\end{equation}
By the same arguments, we prove an estimate similar to \eqref{Est}  for the adjoint operator
$\big(\widetilde{P}_{\theta_1}(h)\big)^*-\overline{z}$ and 
we conclude that $\widetilde{P}_{\theta_1}(h)-z$ is invertible for every $z\in \Gamma_{cM}$. Hence, $\widetilde{P}_{\theta_1}(h)$ has no spectrum in $\Gamma_{cM}$ and we have the estimate 
\begin{equation}
\Vert (\widetilde{P}_{\theta_1}(h)-z)^{-1} \Vert \leq C \kappa(h)^{-1} ,
\end{equation}
uniformly for $z\in \Gamma_{cM}$.
\end{proof}

\subsection{Proof of Theorem \ref{exponential resonance free region}}

We start by proving the following estimate on the distorted resolvent.
\begin{lemma}\label{Lem} \normalfont
Assume \textbf{($\textbf{Hol}_{\infty}$}) and let $\Vert V_{\infty}\Vert_{N\times N}<\alpha<\beta<+\infty$, $\eta>0$ and $\theta = h\vert \ln h \vert$. Then, there exist a constant $C>0$ and $h_0\in (0,1]$ such that 
\begin{equation}\label{equivalence des normes}
 \Vert (P_{\theta}(h)-z)^{-1}\Vert = \mathcal{O}(e^{C/h}),
\end{equation}
uniformly for $z\in [\alpha,\beta]- i [0, \eta h \vert \ln h\vert]$, $z\notin {\rm Res}\,(P(h))$ and $h\in (0,h_0]$.
\end{lemma}

\begin{proof}
Let $\psi\in C_0^{\infty}(\mathbb R^d;[0,1])$ be such that $\psi(x)=1$ for $\vert x\vert \leq 1$ and $\psi(x)=0$ for $\vert x\vert \geq 2$. We introduce the Schr\"odinger operator 
$$
\mathcal{A}(h) := -h^2 \Delta \cdot I_N + V_r(x), \quad {\rm with} \;\;\; V_r(x) :=  \left(1- \psi\left(\frac{x}{r}\right) \right) V(x), \;\;\; r>0.
$$
Let $a_r(x,\xi):= \vert \xi \vert^2 I_N + V_r(x)$ be the semiclassical symbol of $\mathcal{A}(h)$ and set
$$
\Sigma_{E}^r := \bigcup_{j=1}^N \left\{(x,\xi)\in \mathbb R^{2d};\; \vert \xi\vert^2 + \left(1- \psi\left(\frac{x}{r}\right) \right)\lambda_j(x)=E \right\}, \;\;\; \Sigma_{[\alpha,\beta]}^r:= \bigcup_{E\in [\alpha,\beta]}\Sigma_{E}^r.
$$
where we recall that the $\lambda_j(x)$'s are the eigenvalues of $V(x)$, $x\in \mathbb R^d$. From \textbf{($\textbf{Hol}_{\infty}$}), for $r>0$ large enough, we have $\vert \xi\vert^2 > \alpha/2$ on $\pi_{\xi}(\Sigma_{[\alpha,\beta]}^r)$. Here $\pi_{\xi}$ denotes the spatial projection $(x,\xi)\mapsto \xi$. On the other hand, using that $x\cdot \nabla_x V(x)\rightarrow +\infty$ as $\vert x\vert\rightarrow + \infty$ according to \textbf{($\textbf{Hol}_{\infty}$}) and the Cauchy formula, we get 
\begin{eqnarray*}
\{a_r,x\cdot \xi\}(x,\xi) &=& 2 \vert \xi\vert^2 I_N + \frac{x}{r} \nabla_x \psi\left(\frac{x}{r}\right) V(x) - \left(1- \psi\left(\frac{x}{r}\right) \right) x \nabla_x V(x) \\
&\geq & \alpha/2,
\end{eqnarray*}
for $(x,\xi)\in \Sigma_{[\alpha,\beta]}^r$ and $r>0$ large enough. Thus, we deduce that $\mathbb R^{2d}\ni (x,\xi)\mapsto x\cdot \xi$ is an escape function associated with $a_r$ on $\Sigma_{[\alpha,\beta]}^r$ for $r>0$ large enough.  

Let $\mathcal{A}_{\theta}(h)$ be the distorted operator associated with $\mathcal{A}(h)$, obtained by replacing $P(h)$ with $\mathcal{A}(h)$ in \eqref{Distorted operator}. Let $\chi\in C_0^{\infty}(\mathbb R^d;[0,1])$ such that $\chi(x)=1$ for $\vert x \vert <2r$. In particular, we have 
$$
\chi(V-V_r)= \chi \psi_r V = (V-V_r).
$$ 
Using the resolvent identity, we get for $z\notin {\rm Res}\, (P(h))$,
\begin{align*}
(P_{\theta}(h)-z)^{-1} =& (\mathcal{A}_{\theta}(h)-z)^{-1} - (P_{\theta}(h)-z)^{-1} (V-V_r)  (\mathcal{A}_{\theta}(h)-z)^{-1} \\
=& (\mathcal{A}_{\theta}(h)-z)^{-1} - (\mathcal{A}_{\theta}(h)-z)^{-1} (V-V_r)  (\mathcal{A}_{\theta}(h)-z)^{-1} \\
+& (\mathcal{A}_{\theta}(h)-z)^{-1} (V-V_r)  (P_{\theta}(h)-z)^{-1}  (V-V_r) (\mathcal{A}_{\theta}(h)-z)^{-1} \\
=& (\mathcal{A}_{\theta}(h)-z)^{-1} - (\mathcal{A}_{\theta}(h)-z)^{-1} (V-V_r)  (\mathcal{A}_{\theta}(h)-z)^{-1} \\
+& (\mathcal{A}_{\theta}(h)-z)^{-1} (V-V_r)  \chi(P(h)-z)^{-1}  \chi(V-V_r) (\mathcal{A}_{\theta}(h)-z)^{-1} .
\end{align*}
According to the non-trapping estimate \eqref{distortion non cap}, there exists a constant $C'>0$ such that 
\begin{equation}
\Vert (\mathcal{A}_{\theta}(h)-z)^{-1} \Vert = \mathcal{O}(h^{-C'}),
\end{equation}
uniformly for  $z\in [\alpha,\beta]- i [0, \eta h \vert \ln h\vert]$ and $h>0$ small enough. It follows that 
\begin{equation}\label{estui1}
\Vert (P_{\theta}(h)-z)^{-1} \Vert = \mathcal{O}\big(h^{-2 C'} + h^{-2 C'} \Vert\chi(P(h)-z)^{-1}  \chi \Vert \big),
\end{equation}
uniformly for  $z\in [\alpha,\beta]- i [0, \eta h \vert \ln h\vert]$, $z\notin {\rm Res}\,(P(h))$ and $h>0$ small enough. On the other hand, the weighted estimate \eqref{estimation exponentielle du th} clearly implies the same estimate for the truncated resolvent, that is
\begin{equation}\label{estui2}
 \Vert\chi(P(h)-z)^{-1}  \chi \Vert  = \mathcal{O}(e^{C''/h}),
\end{equation}
for some constant $C''>0$, uniformly for $z\in [\alpha,\beta]- i [0, \eta h \vert \ln h\vert]$ and $h>0$ small enough. Putting together \eqref{estui1} and \eqref{estui2}, we obtain the desired estimate \eqref{equivalence des normes}.
\end{proof}

\noindent
\textbf{End of the proof of Theorem \ref{exponential resonance free region}}. Let $J=[\alpha,\beta]\subset (\Vert V_{\infty}\Vert_{N\times N}, +\infty)$ and let $\theta=h\vert \ln h \vert$. For $z\in \mathbb C$ with ${\rm Re}\, z\in J$, we write 
$$
P_{\theta}(h)-z = (P_{\theta}(h)-{\rm Re}\, z) (I-K(z;h)) \;\;\;\;\; {\rm with}\;\;\;\;   K(z;h) := i {\rm Im }\, z (P_{\theta}(h)-{\rm Re}\, z)^{-1}.
$$
According to Lemma \ref{Lem}, there exists a constant $C>0$ such that for $h$ small enough,
$$
\Vert (P_{\theta}(h)-{\rm Re}\, z)^{-1} \Vert \leq  C e^{C/h}.
$$
It follows that for $h$ small enough, ${\rm Re}\, z\in J$ and $\vert {\rm Im}\, z \vert< \frac{1}{C} e^{-C/h}$, we have 
$$
\Vert K(z;h)\Vert \leq C e^{C/h}\vert {\rm Im}\, z \vert < 1.
$$
Therefore $P_{\theta}(h)-z$ is invertible for $z\in  J -i[0, -\frac{1}{C}e^{-C/h}]$, hence $P(h)$ has no resonances in this region. This ends the proof of Theorem \ref{exponential resonance free region}.
\begin{flushright}
$\square$
\end{flushright}

\section*{Acknowledgments}
The author thanks Gilles Lebeau for valuable discussions about Carleman estimates at Universit\'e de Nice in 2018. The author also thanks Jean-Fran\c cois Bony and Claudio Fern\'andez for many stimulating remarks. The research of the author was supported by  CONICYT FONDECYT Grant No. 3180390.


\begin{thebibliography}{99}


\bibitem[Ag]{Agmon}
\textsc{S. Agmon},
\textit{Spectral properties of Schr\"odinger operators and scattering theory}, Ann., Scuola Norm., sup., Pisa, (4)2 (1975), 151-218

\bibitem[AC]{Aguilar}
\textsc{J. Aguilar, J. M. Combes},
\textit{A class of analytic perturbations for one-body Schr\"odinger Hamiltonians}, Comm. Math. Phys., \textbf{22} (1971) 269-279.

\bibitem[Ash]{Ashida} S.~Ashida\,:
\newblock{\em Molecular predissociation resonances below an energy level crossing.}
\newblock{Asymptot. Anal.} 107 (2018), no. 3-4, 135-167.


\bibitem[As]{Assal1}
\textsc{M. Assal},
\textit{Semiclassical resolvent estimates for Schr\"odinger operators with matrix-valued potentials and applications}, work in progress.

\bibitem[ADF]{Assal2}
\textsc{M. Assal, M. Dimassi, S. Fujii\'e},
\textit{Semiclassical trace formula and spectral shift function for systems via a stationary approach}, Int. Math. Res. Notices, Vol. 2019, No. 4, pp. 1227-1264.




\bibitem[Bu]{Burq1}
\textsc{N. Burq},
\textit{Lower bounds for shape resonances widths of long range Schr\"odinger operators}, American Journal of Mathematics, Volume 124, Number 4, 2002, pp. 677-735.

\bibitem[CV]{Cardoso}
    \textsc{F. Cardoso, G. Vodev},
    \textit{Uniform estimates of the resolvent of the Laplace-Beltrami operator on infinite volume Remannian manifolds II},
    Annals of Henri Poincar\'e, (4)3 (2002), 673-691.

        
    
    
\bibitem[DiSj]{Dimassi}
\textsc{M. Dimassi, J. Sj\"ostrand},
\textit{Spectral asymptotics in the semi-classical limit}, London Mathematical Society, Lecture Note Series \textbf{268}, 1999.

\bibitem[Da]{Datchev}
\textsc{K. Datchev},
\textit{Quantitative limiting absorption principle in the semiclassical limit}, Geom. Funct. Anal. Vol. 24 (2014) 740-747.

\bibitem[DDZ]{Da-Dy-Zw}
\textsc{K. Datchev, S. Dyatlov, M. Zworski},
\textit{Resonances and lower resolvent bounds}, J. of. Spectral Theory, 5(3) (2015), 599-615.

\bibitem[FMW1]{FMW1} 
\textsc{S. Fujii\'e, A. Martinez, T. Watanabe},
   \textit{Widths of resonances at an energy-level crossing I: Elliptic interaction}, J. Diff. Eq. 260 (2016) 4051-4085.
     
     
            \bibitem[FMW2]{FMW2} 
            \textsc{S. Fujii\'e, A. Martinez, T. Watanabe},
   \textit{Widths of resonances at an energy-level crossing II: Vector field interaction}, J. Diff. Eq. 262 (2017) 5880-5895.
    
     
      \bibitem[FMW3]{FMW3} 
      \textsc{S. Fujii\'e, A. Martinez, T. Watanabe},
   \textit{Widths of resonances above an energy-level crossing}, Preprint, arxiv: 1904.12511.
        


\bibitem[GM]{Gerard}
\textsc{C. G\'erard, A. Martinez},
\textit{Principe d'absorption limite pour les op\'erateurs de Schr\"odinger \`a longue port\'ee},
C. R. Acad. sci. Paris \textbf{306} (1988), 121-123.





 
 
\bibitem[HM]{Helffer}
\textsc{B. Helffer, A. Martinez},
\textit{Comparaison entre les diverses notions de r\'esonances}, Helv. Phys. Acta \textbf{60} (1987), no. 8, 992-1003.

 
  \bibitem[Hi]{Hi} 
\textsc{K. Higuchi},
 \textit{Resonances free domain for a system of Schr\"odinger operators with energy-level crossings}, Preprint, arxiv: 1912.10180.
    
\bibitem[Hu]{Hunziker}
\textsc{W. Hunziker},
\textit{Distortion analyticity and molecular resonance curves}, Ann. Inst. Henri Poincare, Phys. Theor. 45, 339-358 (1986).
    
    
   
   
\bibitem[Je1]{Jecko}
\textsc{T. Jecko},
\textit{On the mathematical treatment of the Born-Oppenheimer approximation}, Journal of Mathematical Physics 55, (2014). 
  
 \bibitem[Je2]{Jecko2}
\textsc{T. Jecko},
\textit{Estimations de la r\'esolvante pour une mol\'ecule diatomique dans l'approximation de Born-Oppenheimer}, Comm. Math. Phys., 195(3), 585-612., 1998.

  
\bibitem[Je3]{Jecko1}
\textsc{T. Jecko},
\textit{Semiclassical resolvent estimates for Schr\"odinger matrix operators with eigenvalues crossings}, Math. Nachr., 257, 1, p. 36-54, 2003.

 \bibitem[Je4]{Jecko3}
\textsc{T. Jecko},
\textit{Non-trapping condition for semiclassical Schr\"odinger operators with matrix-valued potentials}, Math. Phys. Electronic Journal, 11(2), 2005.



\bibitem[FR]{Clotide}
\textsc{C. Fermanian-Kammerer, V. Rousse},
\textit{Resolvent Estimates and Matrix-Valued Schr\"odinger Operator with Eigenvalue Crossings; Application to Strichartz Estimates},
Comm. in PDE, \textbf{33} (2008), 19-44.   
  
\bibitem[Ka]{Kato}
\textsc{T. Kato},
\textit{Perturbation Theory for Linear Operators}, Berlin: Springer, 1995.  


  
\bibitem[KMSW]{Klein}
\textsc{M. Klein, A. Martinez, R. Seiler, X. P. Wang},
\textit{On the Born-Oppenheimer expansion for polyatomic molecules},
Comm. Math. Phys. \textbf{143} (1992), 607-639.   
   
   
   

\bibitem[Ma]{Martinez1}
\textsc{A. Martinez},
\textit{Resonance Free Domains for Non Globally
Analytic Potentials}, Ann. Henri Poincar\'e \textbf{4} 2002, 736-759.

\bibitem[Mi]{Michel}
\textsc{L. Michel},
\textit{Semi-Classical Behavior of the Scattering Amplitude for Trapping Perturbations at Fixed Energy}, Canad. J. Math., Vol. 56(4), 2004, 794-829.

\bibitem[Na]{Nakamura}
\textsc{S. Nakamura},
\textit{Scattering theory for the shape of resonance model I, II},  Ann. IHP Vol. 50, 1989, 35-52 and 53-62.

\bibitem[NSZ]{Nakamura1}
\textsc{S. Nakamura, P. Stefanov, M. Zworski},
\textit{Resonance expansions of propagators in the presence of potential barriers}, J. Func. Anal., \textbf{205}, (2003) 180-205.

\bibitem[Ne]{Nedelec}
\textsc{L. Nedelec},
\textit{Resonances for matrix Schr\"odinger operators}, Duke Math. J., Vol. 106(2), 2001.

\bibitem[PZ]{Petkov}
\textsc{V. Petkov, M. Zworski},
\textit{Semi-Classical estimates on the scattering determinant}, Ann. Henri Poincar\'e, 2 (2001), 675-711.

\bibitem[RT]{Robert-Tamura}
\textsc{D. Robert, H. Tamura},
\textit{Semiclassical estimates for resolvant and asymptotics for the total scattering cross-sections}, Ann. IHP Phys. Th\'eor., \textbf{46}(4) (1987), 415-442.

\bibitem[RoTa]{Tao}
\textsc{I. Rodianski, T. Tao},
\textit{Effective Limiting Absorption Principles, and Applications}, Commun. Math. Phys. 333, 1-95 (2015).

\bibitem[Sh]{Shapiro}
\textsc{J. Shapiro},
\textit{Semiclassical resolvent bounds in dimension two}, to appear in Proceedings of the American Mathematical Society, arXiv 1604.038452.

\bibitem[Sj]{Sjostrand}
\textsc{J. Sj\"ostrand},
\textit{Lectures on resonances}, \url{http://sjostrand.perso.math.cnrs.fr/Coursgbg.pdf}, 2002.

\bibitem[SZ1]{Sj-Zw}
\textsc{J. Sj\"ostrand, M. Zworski},
\textit{Complex scaling and the distribution of scattering poles}, J. Amer. Math. Soc., \textbf{4} (1991), 729-769.

\bibitem[SZ2]{Sj-Zw1}
\textsc{J. Sj\"ostrand, M. Zworski},
\textit{Fractal upper bounds on the density of semiclassical resonances}, Duke Math. J., \textbf{137}, no. 3 (2007), 381-459.

\bibitem[Vo1]{Vodev}
\textsc{G. Vodev},
\textit{On the exponential bound of the cutoff resolvent}, Serdica Math. J., \textbf{26}, (2000), 49-58.





\bibitem[Zw]{Zworski}
\textsc{M. Zworski},
\textit{Semiclassical Analysis}, Graduate studies in mathematics \textbf{138}, 2012. 


\end{thebibliography}
\end{document}